\documentclass[preprint,amsmath]{revtex4}

\usepackage{bm}

\usepackage{amsmath,amsthm,amssymb,amsfonts,mathrsfs}
\usepackage{amsmath,amssymb,amsthm, mathrsfs}
\usepackage{graphicx}
\usepackage{subfig} 
\usepackage{float}

\newtheorem{thm}{Theorem}[section]

\newtheorem{cor}[thm]{Corollary}

\newtheorem*{ack}{Acknowledgements}

\theoremstyle{definition}

\newtheorem{example}[thm]{Example}

\allowdisplaybreaks[4]

\begin{document}
\title{The complementarity relations in a multi-path interferometer with quantum memory}

\author{Yue Sun$^{1}$, Ming-Jing Zhao$^{2*}$, Peng-Tong Li$^{1}$ \\
        \small $^{1}$School of Mathematics, Nanjing University of Aeronautics and Astronautics, Nanjing 210016, PR China \\
        \small $^{2}$School of Science, Beijing Information Science and Technology University, Beijing 102206, PR China \\
        \small $^{*}$Corresponding author: {zhaomingjingde@126.com} \\
}

\begin{abstract}
The complementarity relations impose the constraints on different aspects of quantum states.
We study the complementarity relation within a multi-path interferometer that includes detectors and quantum memory. Here we consider the mixed states as the input states. We establish a duality relation between the visibility and the path distinguishability. Based on this duality, two triality relations, one is related with visibility, path distinguishability, and
mixedness, the other is related with visibility, path distinguishability, and entanglement, were
derived respectively. Therefore the role of entanglement in multi-path interferometer is characterized quantitatively. These complementarity relations are all complete for the two-path interferometer.

\end{abstract}
 \maketitle
\section{Introduction.}
Complementarity relation is a fundamental law in quantum mechanics. It has profound implications for both theoretical and experimental development of quantum physics. In the
interferometer, the complementarity is manifested by the wave-particle features of the particle. The wave feature is exhibited by the visibility of the interference fringe and the particle feature is exhibited by path information in the interference.
The first analysis of the wave-particle duality in a quantitative way is given by Wootters and Zurek \cite{zuizao}. Then the wave-particle duality is put into a delicate formula by Greenberger and Yasin as \cite{GY}
\begin{equation*}
P^{2}+V^{2}\leq 1,
\end{equation*}
where $P$ stands for the path information of the particle and $V$ stands for the visibility of the interference fringes.
The first quantitative wave-particle duality in multi-path interferometers was established by D\"{u}rr \cite{durr}. In his work, he introduced the criteria for wave measures and particle measures, and provided specific expressions through the moments of functions \cite{durr}. Since then, much attention has been paid to characterize the wave-particle duality  quantitatively \cite{n1,n2,n5,TQ}. Especially the emergence of quantum coherence resource theory promotes further the development of the complementarity relation in multi-path interference. Different types of wave-particle duality are proposed from different perspectives \cite{n6,GR,EK,hh1}. To explore the gap between the unity and the wave-particle features, the mixedness \cite{n3,n2}, entanglement \cite{JC1,JC2,JC3}, and entropy \cite{s1,s2,s3} are involved and the corresponding trialities emerges in the complementarity relation.
The complementarity relations in form of inequalities are called incomplete, while those in form of identities are called complete \cite{Basso20,Basso21}.

To detect the path information of the particle,
the interference affiliated with path detectors is introduced and a wave-particle duality relation is derived as \cite{Englert}
\begin{equation*}
D^{2}+V^{2}\leq 1,
\end{equation*}
where $D$, referred to as {\it path distinguishability}, quantifies the particle properties in the two-path interferometer, and $V$, defined in terms of $l_2$ norm of coherence, quantifies the visibility of the interference fringes.
Later, the duality between the path distinguishability and visibility is presented
in the  $n$-path interference with path detectors \cite{Bera}.
This duality is then improved by Ref. \cite{du1} to a tighter one as
\begin{eqnarray}\label{eq ori}
    (P_{s}-\frac{1}{n})^{2}+X^{2}\leq (1-\frac{1}{n})^{2},
\end{eqnarray}
where $P_{s}$ is the optimal success probability to discriminate detector states
and $X$ is the normalized $l_1$ norm coherence. Subsequently, this duality in Eq. (\ref{eq ori}) is improved further by  assisting the particle with a quantum memory \cite{bkf}
\begin{equation}\label{eq bu}
(P^A_{s}-\frac{1}{n})^{2}+X_A^{2}\leq
(1-\frac{1}{n})^{2} + \frac{2(n-1)}{n^{2}}({\rm Tr}{\rho}_{A}^2-{\rm Tr}{\rho}_{AB}^2),
\end{equation}
where subscript A represents the particle we concerned and subscript B represents the quantum memory.
Apparently, the participation of quantum memory improves the duality relation between the visibility and the path distinguishability.

In this work, we mainly study the complementarity relation in the multi-path interferometers with detectors and quantum memory thoroughly. Our aim is to analyze the impact of quantum memory on the complementarity relation in the multi-path interferometers. As a result, we find the entanglement between the referred particle and the quantum memory plays an essential role in the complementarity relation.

To this end, we first
show the duality relation between the visibility and the path distinguishability, which is complete in the two-path interferometers. As a byproduct, we derive the triality among the visibility, the path distinguishability and mixedness. Next we prove
the complementarity relation among the visibility, path distinguishability, and entanglement analytically, which is illustrated by two explicit examples. This triality relation reveals the essential role of entanglement in the constraint on wave and particle features of particle. Here we should emphasize that these  complementarity relations we derived here are all for the multi-path interferometers with detectors and quantum memory.

\section{Preliminaries}

In this section, we introduce $n$-path interferometer equipped with path detectors and quantum memory, as depicted in Figure \ref{fg1}.
\begin{figure}[htb]
\centering
    \includegraphics[width=\textwidth]{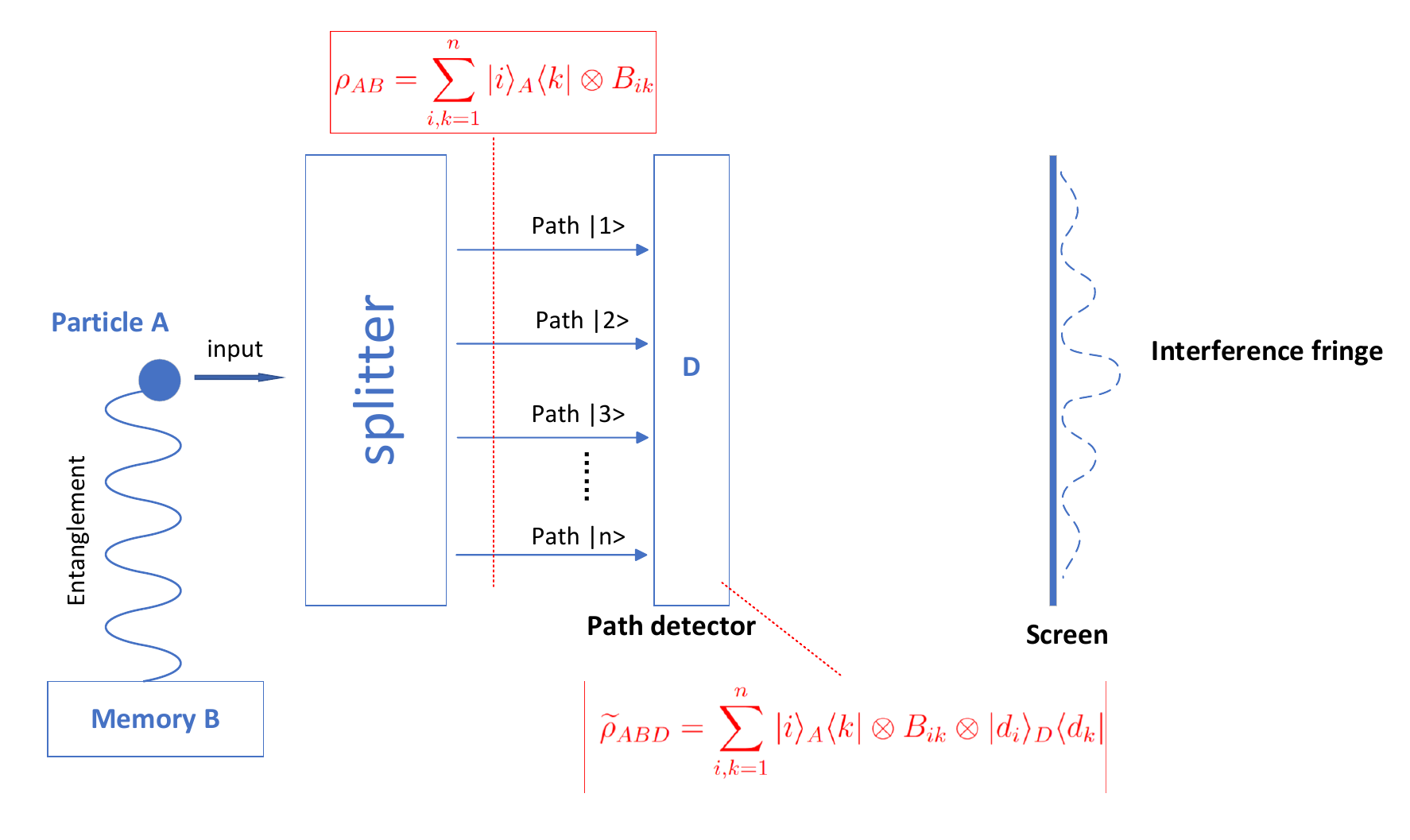}
    \caption{Schematic of $n$-path interferometer with path detector and quantum memory. The particle passes through an $n$-path interferometer while the particle is entangled with a quantum memory. The path detector is used to detect which path a particle passes through. Finally the particles hit the screen to form interference fringes.}
    \label{fg1}
\end{figure}
We suppose $\{|i\rangle_{A}\}_{i=1}^{n}$ is a set of orthonormal basis states of $n$-dimensional Hilbert space denoted by $H_{A}$ with $|i\rangle_{A}$ corresponding to the $i$-th path. Let $\{|e_j\rangle_{B}\}_{j=1}^{n}$ denote another orthonormal basis of the quantum memory, which is an $n$-dimensional Hilbert space denoted by $H_{B}$. Then
the bipartite state $\rho_{AB}$ shared between the particle A and the quantum memory B after passing through the beam splitter can be described as
\begin{equation*}
\rho_{AB}=\sum_{i,k=1}^{n}\sum_{j,l=1}^{n}\rho_{ij,kl}|i\rangle_{A}\langle k|\otimes|e_j\rangle_{B}\langle e_l|,
\end{equation*}
where $\sum_{i,j=1}^{n}\rho_{ij,ij}=1$.
The reduced state of subsystem $A$ is
\begin{equation}\label{before A}
\rho_{A}={\rm Tr}_{B}(\rho_{AB})=\sum_{i,k=1}^{n}\sum_{j=1}^{n}\rho_{ij,kj}|i\rangle_{A}\langle k|.
\end{equation}
Denote the operator
$B_{ik}=\sum_{j,l=1}^{n}\rho_{ij,kl}|e_j\rangle_{B}\langle e_l|$ with ${\rm Tr}(B_{ik})=(\rho_{A})_{ik}$, then the bipartite state $\rho_{AB}$ can also be written as
\begin{equation*}
\rho_{AB}=\sum_{i,k=1}^{n}|i\rangle_{A}\langle k|\otimes B_{ik}.
\end{equation*}
Within the interferometer, the particle A will interact with a system $D$ known as the detector. The detector starts in the state $|d_{0}\rangle$. At this moment, the initial state of the whole system is
\begin{eqnarray*}
{\rho}_{ABD}=\sum_{i,k=1}^{n}|i\rangle_{A}\langle k|\otimes B_{ik}\otimes|d_{0}\rangle_{D}\langle d_{0}|.
\end{eqnarray*}
After the interaction, the state of the whole system becomes
\begin{equation}
\widetilde{\rho}_{ABD}=\sum_{i,k=1}^{n}|i\rangle_{A}\langle k|\otimes B_{ik}\otimes|d_{i}\rangle_{D}\langle d_{k}|,
\end{equation}
where the interaction is described as the controlled unitary operator $U(|i\rangle_{A}\otimes |d_{0}\rangle_{D})=|i\rangle_{A}\otimes|d_{i}\rangle_{D}$.
Tracing out the detector, the density matrix of the composed system $AB$ is
\begin{align*}
\widetilde{\rho}_{AB}&={\rm Tr}_{D}(\widetilde{\rho}_{ABD})
=\sum_{i,k=1}^{n}\langle d_{k}|d_{i}\rangle|i\rangle_{A}\langle k|\otimes B_{ik},
\end{align*}
and the reduced density matrix of the particle A is
\begin{align}\label{y a}
\widetilde{\rho}_{A}&={\rm Tr}_{BD}(\widetilde{\rho}_{ABD})
=\sum_{i,k=1}^{n}\langle d_{k}|d_{i}\rangle (\rho_{A})_{ik}|i\rangle_{A}\langle k|.
\end{align}
Similarly, the density matrix of the detector can be represented as
\begin{equation}\label{eq d}
\widetilde{\rho}_{D}={\rm Tr}_{AB}(\widetilde{\rho}_{ABD})
=\sum_{i=1}^{n}q_{i}|d_{i}\rangle\langle d_{i}|,
\end{equation}
where
\begin{equation}\label{eq qi}
    q_{i}=(\rho_{A})_{ii}.
\end{equation}

In the multi-path interference, the particle A behaves like both wave and particle when it goes through the slits.
To characterize the wave behavior of the particle A in the interference, the $l_2$ norm of coherence is employed as the visibility \cite{durr}. For any quantum state $\rho=\sum_{i,j} \rho_{ij} |i\rangle\langle j|$, its $l_2$ norm of coherence is $C_{l_2}(\rho)=( \sum_{i\neq j} |\rho_{ij}|^2)^{\frac{1}{2}}$. Here we consider a normalized version of the $l_2$ norm of coherence given by $V(\rho)={\frac{\sqrt{2(n-1)}}{n}}C_{l_2}({\rho})$.
So the wave feature exhibited by the quantum state $\widetilde{\rho}_{A}$ before the screen in the interferometer is quantified as
\begin{equation}\label{coherence1}
V(\widetilde{\rho}_{A})^{2}
=\frac{2(n-1)}{n^{2}}\sum_{\substack{i,k=1 \\ i\neq k}}^{n}|\langle d_{k}|d_{i}\rangle|^{2} |(\rho_{A})_{ik}|^{2}.
\end{equation}

On the other aspect, the particle behavior of the particle A in the interference is described by the path information which is encoded within the detector states.
To acquire the path information of the particle A, we need to turn to the detectors and discriminate the detector states $\{|d_{i}\rangle_{D}\}$ with probabilities $\{q_{i}\}$ in Eq. (\ref{eq qi}). Here we utilize the minimum-error state discrimination as the detector states $\{|d_{i}\rangle_{D}\}$ are not necessarily orthogonal.
For a positive operator valued measure $\bf{(POVM)}$
$\Pi=\{\Pi_{i}\}_{i=1}^{n}$ with $\Pi_{i}\geq 0$ and $\sum_{i=1}^{n}\Pi_{i}=\mathbb{I}$, where $\mathbb{I}$ is identity operator,
the success probability to discriminate the states $\{|d_{i}\rangle_{D}\}_{i=1}^{n}$  with probabilities $\{q_{i}\}$ by $\Pi$ is $\sum_{i=1}^{n}q_{i} {\rm Tr}[\Pi_{i}|d_{i}\rangle_{D}\langle d_{i}|]$.
Thus the optimal success probability among all $\bf{POVM}$s is
\begin{equation}\label{hps}
P_{s}(\{q_i, |d_i\rangle_D\})=\max_{\{\Pi_{i}\}}\sum_{i=1}^{n}q_{i} {\rm Tr}[\Pi_{i}|d_{i}\rangle_{D}\langle d_{i}|].
\end{equation}
The upper bound for the
optimal success probability $P_{s}(\{q_i, |d_i\rangle_D\})$ is given by \cite{du1}
\begin{equation}\label{hp_s}
P_{s}(\{q_i, |d_i\rangle_D\})\leq \frac{1}{n}+\frac{1}{2n}\sum_{i,k=1}^{n}||T_{ik}||_{1},
\end{equation}
where the operator $T_{ik}=q_{i}|d_{i}\rangle_{D}\langle d_{i}|-q_{k}|d_{k}\rangle_{D}\langle d_{k}|$.  The $||\cdot||_{1}$ is trace norm and
\begin{equation}\label{hp_s1}
||T_{ik}||_{1}=2\sqrt{(\frac{q_{i}+q_{k}}{2})^{2}-q_{i}q_{k}|\langle d_{i}|d_{k}\rangle|^{2}}.
\end{equation}

The Ref. \cite{du1} indicates that $P_{s}(\{q_i, |d_i\rangle_D\})-\frac{1}{n}$ is an indicator of how much performance can be improved by using prior information and detectors compared to random guessing to acquire the path information. Here we regard it as the path distinguishability to describe the  particle behavior of particle A passing through the interferometer.

\section{ The complementarity relation between visibility and path distinguishability in a multi-path interferometer}

 In this section, we present the complementarity relation between visibility and path distinguishability within an
$n$-path interferometer equipped with a detector and quantum memory.

\begin{thm}\label{th dualtiy}
In the $n$-path interferometer equipped with a detector and quantum memory, the visibility and path distinguishability satisfy
\begin{equation}\label{th1}
(P_{s}(\{q_i, |d_i\rangle_D\})-\frac{1}{n})^{2}+V (\widetilde{\rho}_{A})^{2}\leq
(1-\frac{1}{n})^{2}-\frac{2(n-1)}{n^{2}}({\rm Tr}[\widetilde{\rho}_{D}^{2}]-{\rm Tr}[\widetilde{\rho}_{A}^{2}]).
\end{equation}
\end{thm}

\begin{proof}
In the $n$-path interferometer equipped with a detector and quantum memory, by  Eq. (\ref{hp_s}) and Eq. (\ref{hp_s1}), we get that the path distinguishability satisfies
\begin{align*}
(P_{s}(\{q_i, |d_i\rangle_D\})-\frac{1}{n})^{2}
&\leq\frac{1}{n^{2}}\Bigg(\sum_{\substack{i,k=1 \\ i\neq k}}^{n}\sqrt{(\frac{q_{i}+q_{k}}{2})^{2}-q_{i}q_{k}|\langle d_{i}|d_{k}\rangle|^{2}}\Bigg)^{2}\nonumber\\
&\leq \frac{1}{n^{2}}\Bigg(\sum_{\substack{i,k=1 \\ i\neq k}}^{n}\frac{q_{i}+q_{k}}{2}\Bigg)
\Bigg[\sum_{\substack{i,k=1 \\ i\neq k}}^{n} (\frac{q_{i}+q_{k}}{2}-\frac{2q_{i}q_{k}}{q_{i}+q_{k}}|\langle d_{i}|d_{k}\rangle|^{2})\Bigg],
\end{align*}
where we have used the Cauchy-Schwarz inequality for the last inequality.
Note the relation $\sum_{i\neq k=1}^{n}\frac{q_{i}+q_{k}}{2}=n-1$, the inequality above is then bounded by
\begin{align}\label{cs1}
(P_{s}(\{q_i, |d_i\rangle_D\})-\frac{1}{n})^{2}
&\leq \frac{(n-1)}{n^{2}}\Bigg[(n-1)-\sum_{\substack{i,k=1 \\ i\neq k}}^{n} \frac{2q_{i}q_{k}}{q_{i}+q_{k}}|\langle d_{i}|d_{k}\rangle|^{2}\Bigg]\nonumber\\
&\leq\frac{(n-1)}{n^{2}}\Bigg[(n-1)-2\sum_{\substack{i,k=1 \\ i\neq k}}^{n}q_{i}q_{k}|\langle d_{i}|d_{k}\rangle|^{2}\Bigg],
\end{align}
due to $q_{i}+q_{k}\leq 1$.
Therefore the sum of the path distinguishability and the visibility is limited by
\begin{align*}\label{wp1}
&(P_{s}(\{q_i, |d_i\rangle_D\})-\frac{1}{n})^{2}+V(\widetilde{\rho}_{A})^{2}\nonumber\\
&\leq\frac{(n-1)}{n^{2}}\Bigg((n-1)-2\sum_{\substack{i,k=1 \\ i\neq k}}^{n}q_{i}q_{k}|\langle d_{i}|d_{k}\rangle|^{2}\Bigg)+\frac{2(n-1)}{n^{2}}\Bigg(\sum_{\substack{i,k=1 \\ i\neq k}}^{n}|\langle d_{i}|d_{k}\rangle|^{2}|(\rho_{A})_{ik}|^{2}\Bigg)\nonumber\\
&=(1-\frac{1}{n})^{2}-\frac{2(n-1)}{n^{2}}({\rm Tr}[\widetilde{\rho}_{D}^{2}]-{\rm Tr}[\widetilde{\rho}_{A}^{2}]),
\end{align*}
which deduces the visibility and path distinguishability duality in Eq. (\ref{th1}).
\end{proof}

Theorem \ref{th dualtiy} shows the visibility and path distinguishability duality in the $n$-path interferometer equipped with a detector and quantum memory. For the $n$-path interferometer equipped with a detector but without quantum memory, if the initial state of the particle is a pure state, analogously to the proof of Theorem \ref{th dualtiy}, one can derive that the visibility and path distinguishability satisfy
\begin{eqnarray}\label{no memory}
(P_{s}(\{q_i, |d_i\rangle_D\})-\frac{1}{n})^{2}+V (\widetilde{\rho}_{A})^{2}\leq
(1-\frac{1}{n})^{2}.
\end{eqnarray}
Especially for the two-path case, Eq. (\ref{no memory}) turns into the
 duality relation presented in \cite{Englert}. By comparing Eq. (\ref{th1}) and Eq. (\ref{no memory}), it is easy to find the $n$-path interferometer with quantum memory indeed improves the duality between the visibility and path distinguishability.

Now we make a comparison between the visibility and path distinguishability dualities in Theorem \ref{th dualtiy} and Eq. (\ref{eq bu}) proved in Ref. \cite{bkf}, which share some similarities.
First, the duality in Eq. (\ref{eq bu}) refers to pure state as the input state and ours in Eq. (\ref{th1}) is for mixed state $\rho_{AB}$.
Although all mixed states can be purified theoretically, the purification of an unknown mixed state is in fact terribly hard. In view of the fact that the pure state is rather fragile and it is unavoidably affected by the environment and gets mixed.
So  the duality for mixed states is more general and is meaningful.
Second, it is straightforward to verify that for pure input states, the upper bound in Eq. (12) is consistent with that in Eq. (2), as the relation $\rm{Tr}[\tilde{\rho}_{AB}^{2}] = \rm{Tr}[\tilde{\rho}_{D}^{2}]$ holds. But this is not true for mixed input states. So the upper bound at the right hand side of Eq. (12) and Eq. (2) are different. Third, the $l_{2}$-norm coherence  is originally the second moment of the interference pattern \cite{durr} and it characterizes the  visibility by its own way compared with other measures.
    Finally,  the duality in Theorem \ref{th dualtiy} reduces to  Eq. (\ref{eq bu}) for pure input state in two-path interferometer. But in multi-path interferometer, our result shows its advantages over Eq. (\ref{eq bu}) even for some pure states. For instance,  consider the pure state $|\psi\rangle_{AB}=\sum_{i=1}^{3}\sqrt{p_{i}}|i\rangle_{A}|u_{i}\rangle_{B}$ as the input state in three-path interferometer. Here, $p_{i}=\sum_{j=1}^{3}|a_{ij}|^{2}$ with $\sum_{i=1}^{3}p_{i}=1$ and the state $|u_{i}\rangle_{B}$ is defined as $|u_{i}\rangle_{B}=\sum_{j=1}^{3}a_{ij}|e_j\rangle_{B}/\sqrt{p_{i}}$. If we take the coefficients $a_{11}=\sqrt{\frac{p}{{3}}}, a_{21}=\sqrt{\frac{q}{{3}}}, a_{33}=\sqrt{\frac{3-p-q}{{3}}}$ and the rest are zero, where $0\leq p, q\leq 1$. Then
\begin{equation}
    |\psi\rangle_{AB}=\sqrt{\frac{p}{{3}}}|1\rangle |e_1\rangle+\sqrt{\frac{q}{{3}}}|2\rangle|e_1\rangle +\sqrt{\frac{3-p-q}{{3}}}|3\rangle |e_3\rangle.
\end{equation}
We specify the detector states $\{|d_{i}\rangle\}$ satisfy
$\langle d_{2}|d_{1}\rangle=\frac{1}{3}$, then we can get the reduced density matrix of particle A after the slits is
\[
\begin{gathered}
\mathbf{\widetilde{\rho}}_{A}=
\begin{pmatrix}
\frac{p}{3} & \frac{\sqrt{pq}}{9} & 0 \\
\frac{\sqrt{pq}}{9} & \frac{q}{3} & 0 \\
0 & 0 & \frac{3-p-q}{3} \\
\end{pmatrix}
\end{gathered}.
\]
By calculation we find that the visibility defined in Eq. (\ref{eq bu}) is $X_{A}^{2}=\frac{4pq}{3^{6}}$, while the visibility in Eq. (\ref{th1}) is $V (\widetilde{\rho}_{A})^{2}=\frac{8pq}{3^{6}}$. Taking into account that the path distinguishabilities in the left hand side and the upper bound in the right hand side of the dualities in Eq. (\ref{eq bu}) and Eq. (\ref{th1}) are the same, but the visibility $V (\widetilde{\rho}_{A})$ in Eq. (\ref{th1}) is strictly greater than that $X_{A}$ in Eq. (\ref{eq bu}),
so the duality  (\ref{th1}) is tighter than that (\ref{eq bu}) for this case.

Especially, when $n=2$,  the optimal success probability $P_{s}(\{q_i, |d_i\rangle_D\})$ in Eq. (\ref{hp_s}) reaches its upper bound \cite{HC}. In this case, the duality in  Eq. (\ref{th1}) is complete.

\begin{cor}\label{cor n2}
In the two-path interferometer equipped with a detector and quantum memory, the visibility and path distinguishability duality is
\begin{equation}\label{eq n2}
(P_{s}(\{q_i, |d_i\rangle_D\})-\frac{1}{2})^{2}+V(\widetilde{\rho}_{A})^{2}=
\frac{1}{4}-\frac{1}{2}({\rm Tr}(\widetilde{\rho}^2_{D})-{\rm Tr}(\widetilde{\rho}^2_{A})).
\end{equation}
\end{cor}

Furthermore, for any quantum state $\rho$, if we choose the normalized mixedness as $M(\rho)=\frac{2(n-1)}{n^{2}}(1- {\rm Tr} \rho^2)$,  then by moving the purity ${\rm Tr}(\widetilde{\rho}^2_{A})$ from the right hand side to the left hand side in Eq. (\ref{th1}), the duality between the visibility and path distinguishability is transformed into the triality among the visibility, path distinguishability and mixedness.
\begin{thm}
In the $n$-path interferometer equipped with a detector and quantum memory, the visibility, path distinguishability and mixedness satisfy
\begin{equation}
(P_{s}(\{q_i, |d_i\rangle_D\})-\frac{1}{n})^{2}+V(\widetilde{\rho}_{A})^{2}+M(\widetilde{\rho}_{A})\leq (1-\frac{1}{n})^{2}+M(\widetilde{\rho}_{D}).
\end{equation}
\end{thm}

In particular, when $n=2$, one can derive a complete triality relation about the visibility, path distinguishability and mixedness.

\begin{cor}
In the two-path interferometer equipped with a detector and quantum memory,
if the detector states satisfy $|\langle d_{1}|d_{2}\rangle|^{2}=\frac{1}{2}$,  then
the visibility, path distinguishability and mixedness satisfy the complete triality relation as
\begin{equation}\label{eq tri m}
(P_{s}(\{q_i, |d_i\rangle_D\})-\frac{1}{2})^{2}+\frac{1}{2}V(\widetilde{\rho}_{A})^{2}+\frac{1}{2}M(\widetilde{\rho}_{A})= \frac{1}{4}.
\end{equation}
\end{cor}

\begin{proof}
In the two-path interferometer equipped with a detector and quantum memory, if the detector states satisfy $|\langle d_{1}|d_{2}\rangle|^{2}=\frac{1}{2}$, then the optimal success probability reduces to
\begin{equation*}
(P_{s}(\{q_i, |d_i\rangle_D\})-\frac{1}{2})^{2}=\frac{1}{4}-q_{1}q_{2}|\langle d_{1}|d_{2}\rangle|^{2}=\frac{1}{4}-\frac{1}{2}q_{1}q_{2}.
\end{equation*}
Along with the expressions of the visibility
\begin{equation*}
    V(\widetilde{\rho}_{A})^{2}=\frac{1}{2}|(\rho_{A})_{12}|^{2}
\end{equation*}
and the mixedness of
\begin{equation*}
    M(\widetilde{\rho}_{A})=\frac{1}{2}(1-|(\rho_{A})_{12}|^{2}-\sum_{i=1}^{2}q_{i}^{2}),
\end{equation*}
it is obvious to derive the complete triality in Eq. (\ref{eq tri m}).
\end{proof}

The duality in Eq. (\ref{eq n2}) and triality in Eq. (\ref{eq tri m}) we obtained above are both complete for the two-path interferometer. These impose the constraints on the visibility, path distinguishability and mixedness quantitatively. For example, Eq. (\ref{eq tri m}) tells us that the sum of the visibility and path distinguishability increases as the mixedness of $\widetilde{\rho}_{A}$ decreases.
In contrast, if the output state $\widetilde{\rho}_{A}$ gets more mixed, then the sum of the visibility and path distinguishability gets less and the wave-particle feature of particle A is transformed partly into its mixedness.

\section{The triality relation among the visibility, path distinguishability, and entanglement in a multi-path interferometer}

In this section we explore the role of entanglement in the $n$-path interferometer and establish the triality relation among
the visibility, path distinguishability, and entanglement. Here we utilize the concurrence as the entanglement measure \cite{en}. For any pure state $|\psi\rangle_{AB}=\sum_{i,j} \psi_{ij}|ij\rangle_{AB}$, its concurrence is defined as
$C(|\psi\rangle_{AB})=\sqrt{2(1-{\rm Tr} \rho_A^2)}$,
with $\rho_A={\rm Tr}_B (|\psi\rangle_{AB}\langle\psi|)$. For any mixed state $\rho_{AB}$, its concurrence is defined with the convex roof construction $C(\rho_{AB})=\min_{\{p_i, |\psi_i\rangle_{AB}\}} \sum_i p_i C(|\psi_i\rangle_{AB})$ with the minimum running over all possible pure-state decompositions $\rho_{AB}=\sum_i p_i |\psi_i\rangle_{AB}\langle\psi_i|$. Here we use a normalized version of concurrence as $E(\rho_{AB})=\frac{\sqrt{n-1}}{n} C(\rho_{AB})$. 

\begin{thm}\label{eq vde}
In the $n$-path interferometer equipped with a detector and quantum memory,
the visibility, path distinguishability and entanglement satisfy
\begin{equation}\label{eq vde1}
(P_{s}(\{q_i, |d_i\rangle_D\})-\frac{1}{n})^{2}+V(\widetilde{\rho}_{A})^{2}+ E (\rho_{AB})^{2}\leq \frac{n^{2}-1}{n^{2}}-\frac{2(n-1)}{n^{2}} {\rm Tr}(\widetilde{\rho}^2_{AB}).
\end{equation}
\end{thm}

\begin{proof}
First we assume the quantum state ${\rho}_{AB}$ shared between particle A and quantum memory B is pure,
\begin{equation}\label{eq pure in}
|\psi\rangle_{AB}=\sum_{i=1}^{n}\sqrt{p_i}|i\rangle_{A}|u_{i}\rangle_{B},
\end{equation}
where $p_{i}=\sum_{j=1}^{n}|a_{ij}|^{2}$ and $|u_{i}\rangle_{B}=\sum_{j=1}^{n}a_{ij}|e_j\rangle_{B}/\sqrt{p_{i}}$. Here $|u_{i}\rangle_{B}$ are normalized but not necessarily orthogonal.
After the interaction with the detectors by the controlled unitary $U$, the whole quantum state of particle A, quantum memory B and detectors D are
\begin{equation*}
|\widetilde{\psi}\rangle_{ABD}=\sum_{i=1}^{n}\sqrt{p_i}|i\rangle_{A}|u_{i}\rangle_{B} |d_i\rangle_D.
\end{equation*}
The reduced density matrix of system $AB$ becomes
\begin{equation*}
\widetilde{\rho}_{AB}=\sum_{i,j=1}^{n}\sqrt{p_{i}p_{j}}\langle d_{j}|d_{i}\rangle|i\rangle_{A}\langle j|\otimes|u_{i}\rangle_{B}\langle u_{j}|,
\end{equation*}
and the reduced density matrix of the particle $A$ becomes
\begin{equation*}
\widetilde{\rho}_{A}=\sum_{i,j=1}^{n}\sqrt{p_{i}p_{j}}\langle d_{j}|d_{i}\rangle\langle u_{j}|u_{i}\rangle|i\rangle_{A}\langle j|.
\end{equation*}
Moreover, the density matrix of detector is
\begin{equation*}
   \widetilde{\rho}_{D}=\sum_{i=1}^{n}p_{i}|d_{i}\rangle\langle d_{i}|.
\end{equation*}
The entanglement of the initial state is
\begin{eqnarray}\label{jc}
  E(|\psi\rangle_{AB})^{2}
  = \frac{2(n-1)}{n^{2}}(1 -\sum_{i,j=1}^{n}p_{i}p_{j}|\langle u_{j}|u_{i}\rangle|^{2}).
\end{eqnarray}
Note that ${\rm Tr}[(\widetilde{\rho}_{AB})^{2}]=\sum_{i,j=1}^{n}p_{i}p_{j}|\langle d_{j}|d_{i}\rangle|^{2}$,
then the visibility, path distinguishability and entanglement satisfy
\begin{align*}
&(P_{s}^{A}(\{p_i, |d_i\rangle_D\})-\frac{1}{n})^{2}+V(\widetilde{\rho}_{A})^{2}+E(|\psi\rangle_{AB})^{2}\nonumber\\
&\leq (1-\frac{1}{n})^{2}-\frac{2(n-1)}{n^{2}}\sum_{\substack{i,j=1 \\ i\neq j}}^{n}p_{i}p_{j}|\langle d_{j}|d_{i}\rangle|^{2}+\frac{2(n-1)}{n^{2}}\sum_{\substack{i,j=1 \\ i\neq j}}^{n}p_{i}p_{j}|\langle d_{j}|d_{i}\rangle|^{2}|\langle u_{j}|u_{i}\rangle|^{2}\nonumber\\&+\frac{2(n-1)}{n^{2}}(1-\sum_{i,j=1}^{n}p_{i}p_{j}|\langle u_{j}|u_{i}\rangle|^{2})\nonumber\\
&\leq\frac{n^{2}-1}{n^{2}}-\frac{2(n-1)}{n^{2}}\sum_{\substack{i,j=1 \\ i\neq j}}^{n}p_{i}p_{j}|\langle d_{j}|d_{i}\rangle|^{2}+\frac{2(n-1)}{n^{2}}\sum_{\substack{i,j=1 \\ i\neq j}}^{n}p_{i}p_{j}|\langle u_{j}|u_{i}\rangle|^{2}\nonumber\\&-\frac{2(n-1)}{n^{2}}\sum_{i,j=1}^{n}p_{i}p_{j}|\langle u_{j}|u_{i}\rangle|^{2}\nonumber\\
&=\frac{n^{2}-1}{n^{2}}-\frac{2(n-1)}{n^{2}} {\rm Tr}[(\widetilde{\rho}_{AB})^{2}],
\end{align*}
where the first inequality is based on Eq. (\ref{cs1}) and the second inequality arises from the fact that $|\langle d_{j}|d_{i}\rangle|\leq 1$, for all $i,j.$ 
So the relation among visibility, path distinguishability and entanglement of pure state $|\psi\rangle_{AB}$ satisfy
\begin{equation}\label{pf pure}
(P_{s}(\{p_i, |d_i\rangle_D\})-\frac{1}{n})^{2}+V(\widetilde{\rho}_{A})^{2}+ E(|\psi\rangle_{AB})^{2} \leq \frac{n^{2}-1}{n^{2}}-\frac{2(n-1)}{n^{2}} {\rm Tr}(\widetilde{\rho}^2_{AB}).
\end{equation}

Second, we assume the quantum state ${\rho}_{AB}$ shared between particle A and quantum memory B is  mixed, $\rho_{AB}=\sum_s r_s |\psi_s\rangle_{AB}\langle \psi_s|$ with $0\leq r_s \leq 1$, $\sum_s r_s=1$, $|\psi_s\rangle_{AB}=\sum_{i,j} a_{ij}^s |i\rangle_A |e_j\rangle_{B}$. For pure state $|\psi_s\rangle_{AB}$, suppose the reduced state of subsystem A is $\rho_A^{\psi_{s}}$. Furthermore, we suppose the optimal success probability in the interferometer with respect to  $|\psi_s\rangle_{AB}$ is $P_s(\{q_i^{s}, |d_i\rangle_D\})$, that is,
\begin{eqnarray*}
P_s(\{q_i^{s}, |d_i\rangle_D\})=\max_{\Pi} \sum_{i=1}^{n} (\rho_A^{\psi_{s}})_{ii} {\rm Tr}[\Pi_i |d_i\rangle\langle d_i|],
\end{eqnarray*}
with $q_i^{s}=(\rho_A^{\psi_{s}})_{ii}$.
Then the optimal success probability $P_s(\{q_i, |d_i\rangle_D\})$ with respect to the  mixed state $\rho_{AB}$ satisfies
\begin{eqnarray*}
P_s(\{q_i, |d_i\rangle_D\})&=&\max_{\Pi} \sum_i (\rho_A)_{ii} {\rm Tr}[\Pi_i |d_i\rangle\langle d_i|]\\
&=&\max_{\Pi} \sum_i (\sum_s r_s \rho_A^{\psi_{s}})_{ii} {\rm Tr}[\Pi_i |d_i\rangle\langle d_i|]\\
&\leq&\sum_s r_s \max_{\Pi} \sum_i (\rho_A^{\psi_{s}})_{ii} {\rm  Tr}[\Pi_i |d_i\rangle\langle d_i|]\\
&=&\sum_s r_s P_s(\{q_i^{s}, |d_i\rangle_D\}),
\end{eqnarray*}
with $\rho_A={\rm Tr}_B (\rho_{AB})$ and $q_i=(\rho_A)_{ii}$.
This implies
\begin{eqnarray*}
(P_s(\{q_i, |d_i\rangle_D\})-\frac{1}{n})^2
&\leq&(\sum_s r_s P_s(\{q_i^{s}, |d_i\rangle_D\})-\frac{1}{n})^2\\
&\leq&\sum_s r_s [P_s(\{q_i^{s}, |d_i\rangle_D\})-\frac{1}{n}]^2,
\end{eqnarray*}
where we have used the convexity of the square function.
Subsequently, the sum of path-distinguishability, visibility and entanglement is
\begin{eqnarray*}
&&(P_{s}(\{q_i, |d_i\rangle_D\})-\frac{1}{n})^{2}+V(\widetilde{\rho}_{A})^{2}+ E(\rho_{AB})^{2} \\&\leq&
\sum_s r_s [(P_s(\{q_i^{s}, |d_i\rangle_D\})-\frac{1}{n})^{2}+V(\widetilde{\rho}_{AB}^{\psi_{s}})^{2}+ E(|\psi_s\rangle_{AB})^{2} ]\\
 &\leq& \frac{n^{2}-1}{n^{2}}-\frac{2(n-1)}{n^{2}} \sum_s r_s {\rm Tr}[(\widetilde{\rho}_{AB}^{\psi_{s}})^2]\\
  &\leq&  \frac{n^{2}-1}{n^{2}}-\frac{2(n-1)}{n^{2}} {\rm Tr}(\widetilde{\rho}^2_{AB}),
\end{eqnarray*}
where we have used the convexity of path-distinguishability, concurrence, visibility $V(\rho)$ in the first inequality,
Eq. (\ref{pf pure}) in the second inequality, and the convexity of purity ${\rm Tr} (\rho^2)$ in the third inequality, respectively.
Finally this completes the proof.
\end{proof}

Theorem \ref{eq vde} shows the complementarity relation among visibility, path distinguishability, and entanglement. It demonstrates quantitatively that the constraint of entanglement shared between the referred particle and quantum memory on the duality between the visibility and path distinguishability.

In particular, we consider the case with pure initial state $|\psi\rangle_{AB}$ in  the $n$-path interferometer. We set the detector states $\{|d_i\rangle\}$ satisfy $|\langle d_j|d_i\rangle|^2=\frac{1}{{d}}$ or $|\langle d_j|d_i\rangle|^2=0$ for different $i$ and $j$ for simplicity.
For this case,
the entanglement of the initial pure state is
\begin{eqnarray}\label{pve pure-e}
    E(|\psi\rangle_{AB})^{2}=\frac{2(n-1)}{n^{2}}(1-\sum_{i,j=1}^{n}p_{i}p_{j}|\langle u_{j}|u_{i}\rangle|^{2}),
\end{eqnarray}
and the visibility is
\begin{equation}\label{pve pure-v}
    V(\widetilde{\rho}_{A})^{2}\leq\frac{2(n-1)}{ n^{2}} \frac{1}{d} \sum_{\substack{i,j=1 \\ i\neq j}}^{n}p_{i}p_{j}|\langle u_{j}|u_{i}\rangle|^{2}.
\end{equation}
These two equations above demonstrate a tradeoff relation between entanglement and visibility.
When
$|\langle u_j|u_i\rangle|^2$ approaches 1 for all $i$ and $j$, $E(|\psi\rangle_{AB})$ in Eq. (\ref{pve pure-e}) decreases to 0 while $V(\widetilde{\rho}_{A})$ in Eq. (\ref{pve pure-v}) increases.
This gives rise to a duality relation between path information and visibility as
\begin{equation}
(P_{s}(\{p_i, |d_i\rangle_D\})-\frac{1}{n})^{2}+V(\widetilde{\rho}_{A})^{2}\leq (1-\frac{1}{n})^{2},
\end{equation}
which is the wave-particle duality in the $n$-path interferometer without quantum memory because the quantum memory is separated from the particle A.
Especially, if $n=2$, the duality relation above is complete,
\begin{equation}
(P_{s}(\{p_i, |d_i\rangle_D\})-\frac{1}{n})^{2}+V(\widetilde{\rho}_{A})^{2}= \frac{1}{4}.
\end{equation}
When $|\langle u_j|u_i\rangle|^2$  approaches 0 for all different $i$ and $j$, which implies that $\{|u_{i}\rangle\}$ is an orthonormal basis, then the entanglement $E(|\psi\rangle_{AB})$ in Eq. (\ref{pve pure-e}) increases to its maximum value $\frac{2(n-1)^{2}}{n^{3}}$ and visibility $V(\widetilde{\rho}_{A})$  in Eq. (\ref{pve pure-v}) decreases to its minimum value 0 .
This gives rise to a duality relation between entanglement and path information as
\begin{equation}
(P_{s}(\{p_i, |d_i\rangle_D\})-\frac{1}{n})^{2}+ E (|\psi\rangle_{AB})^{2}\leq \frac{n^{2}-1}{n^{2}}-\frac{2(n-1)}{n^{2}} {\rm Tr}(\widetilde{\rho}^2_{AB}).
\end{equation}
Especially, if $n=2$, the duality relation above is complete,
\begin{equation}
(P_{s}(\{p_i, |d_i\rangle_D\})-\frac{1}{n})^{2}+ E (|\psi\rangle_{AB})^{2}= \frac{3}{4}-\frac{1}{2} {\rm Tr}(\widetilde{\rho}^2_{AB}).
\end{equation}
Next we illustrate the complementarity relation in Eq. (\ref{eq vde1}) in two-path interferometer specifically.

\begin{example}
In a two-path interferometer, if the quantum state shared between the referred particle and quantum memory is pure bipartite state
\begin{eqnarray}\label{exa pure}
    |\psi\rangle_{AB}=\sqrt{p}|1\rangle_{A}|u_{1}\rangle_{B}+\sqrt{1-p}|2\rangle_{A}|u_{2}\rangle_{B},
\end{eqnarray}
then the path distinguishability is
\begin{equation*}
    (P_s(\{p_{i},|d_{i}\rangle\})-\frac{1}{2})^{2}=\frac{1}{4}-p(1-p)|\langle d_{1}|d_{2}\rangle|^{2},
\end{equation*}
which is independent of the basis of the quantum memory,
the visibility is
\begin{equation*}
    V(\widetilde{\rho}_{A})^{2}=p(1-p)|\langle d_{1}|d_{2}\rangle|^{2}|\langle u_{1}|u_{2}\rangle|^{2}
\end{equation*}
and the entanglement of $|\psi\rangle_{AB}$ is
\begin{equation*}
    E(|\psi\rangle_{AB})^{2}=p(1-p)(1-|\langle u_{1}|u_{2}\rangle|^{2})
\end{equation*}
respectively.
To express more concisely, we denote $|\langle u_{1}|u_{2}\rangle|^{2}$ by $c_u$ and $|\langle d_{1}|d_{2}\rangle|^{2}$ by $c_{d}$.
Thus, the left hand side of Eq. (\ref{eq vde1}) is
\begin{equation}\label{ex l}
    \frac{1}{4}+p(1-p)(c_{u}-1)(c_{d}-1)
\end{equation}
and
the right hand side of the Eq. (\ref{eq vde1}) is
\begin{equation}\label{ex r}
    \frac{1}{4}+p(1-p)(1-c_{d}).
\end{equation}

Now we consider the effect of $c_{u}$ on the triality in Eq. (\ref{eq vde1}).
When $c_{u}=0$ which means  $|u_1\rangle$ and $|u_2\rangle$ are orthogonal, then Eq. (\ref{ex l}) coincides with Eq. (\ref{ex r}). At
this time, the entanglement shared between the referred particle and quantum memory reaches its maximum with fixed $p$ and the visibility of $\widetilde{\rho}_{A}$ vanishes.
If we demand that $c_{d}=0$ and $p=0.5$ additionally, then Eqs. (\ref{ex l}) and Eq. (\ref{ex r}) reach the maximum $\frac{1}{2}$.
When $c_{u}=1$ which means  $|u_1\rangle$ and $|u_2\rangle$ are the same up to a global phase, the entanglement shared between the referred particle and quantum memory vanishes and the visibility of $\widetilde{\rho}_{A}$ reaches its maximum for fixed $p$ and detector states.
The estimation of the triality in Eq. (\ref{eq vde1}) for pure state in Eq. (\ref{exa pure}) is plotted in FIG. \ref{tu1} and FIG. \ref{tu}. In FIG. \ref{tu1}, the surfaces are the plots of Eq. (\ref{ex l}) with respect to parameters $p$ and $c_d$ for fixed $c_u$ ($c_u=0,\ 0.5,\ 0.9$ from the top to the bottom). The top surface is also the plot of Eq. (\ref{ex r}) because Eqs. (\ref{ex l}) and (\ref{ex r}) coincide for $c_u=0$.

In FIG \ref{tu}, we fix $c_{d}$ at 1/2 and examine the effect of parameters $p$ and $c_{u}$ on the triality relation. The subfigure (a) is the plot of the Eq. (\ref{ex l}), which decreases as parameter $c_{u}$ increases. The subfigure (b) is the plot of Eq. (\ref{ex l}) about the parameter $p$ for fixed $c_{u}$.
It is symmetric and convex with respect to $p$ and reaches its maximum value at $p = 1/2$.
The subfigure (c) is the plot of Eq. (\ref{ex l}) about the parameter $c_{u}$ for fixed $p$. It decreases linearly as $c_{u}$ increases.
\begin{figure}[ht]
    \begin{minipage}{1\textwidth}
      \centering
       \includegraphics[width=\textwidth]{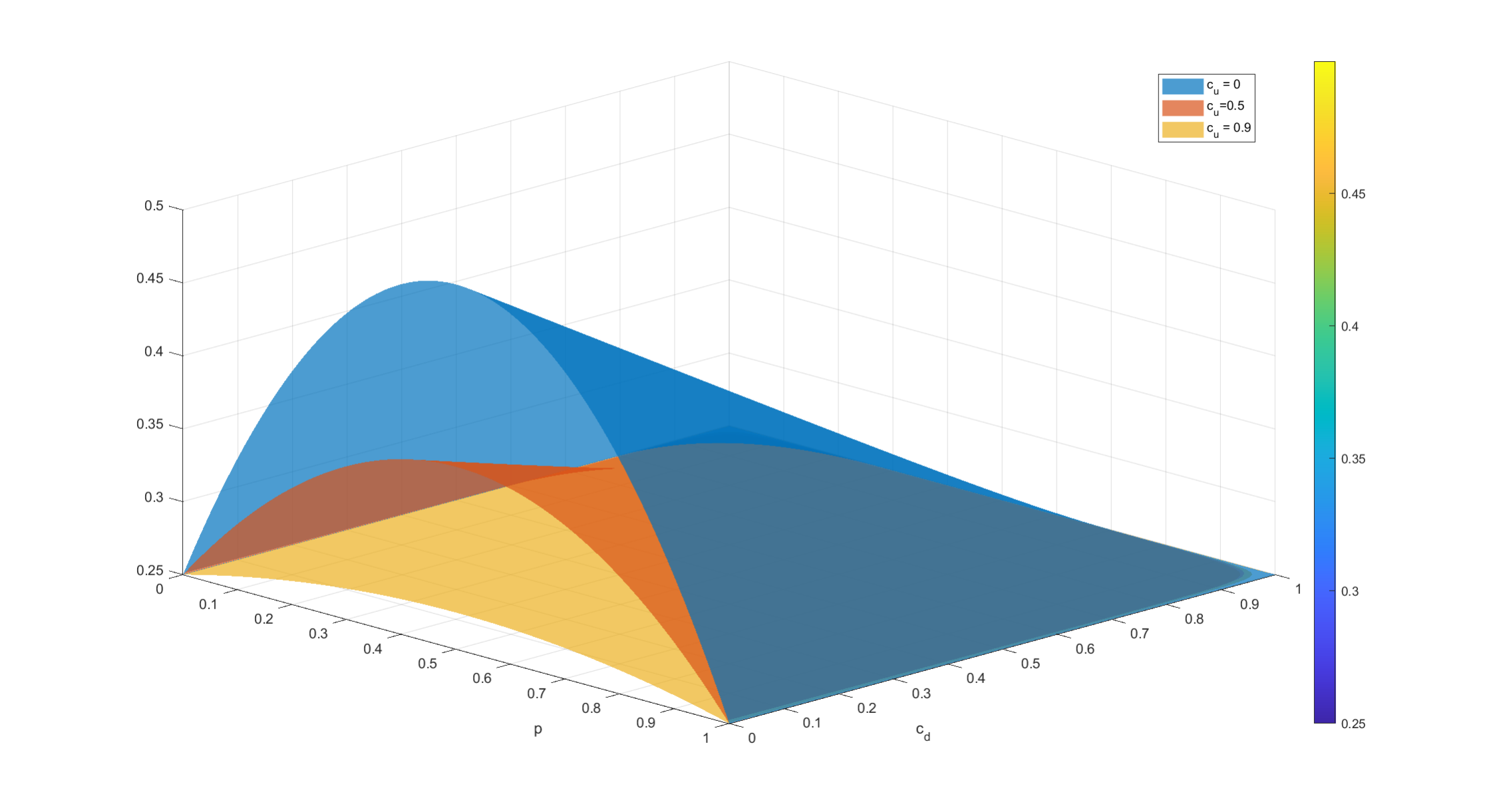}
   \end{minipage}\hfill
    \caption{(Color online) The estimation of the triality relation among visibility, path distinguishability, and entanglement with respect to the parameter $p$ and $|\langle d_1|d_2\rangle|^2=c_{d}$.
    The three surfaces from top to bottom  correspond to the function of Eq.  (\ref{ex l}) when $c_{u}=|\langle u_1|u_2\rangle|^2$ takes values 0, 0.5, 0.9 respectively. The surface on the top is also the plot of the Eq.  (\ref{ex r}).}
    \label{tu1}
\end{figure}

\begin{figure}[H]
  \centering
  \begin{minipage}{0.5\textwidth}
    \centering
    \includegraphics[width=\linewidth]{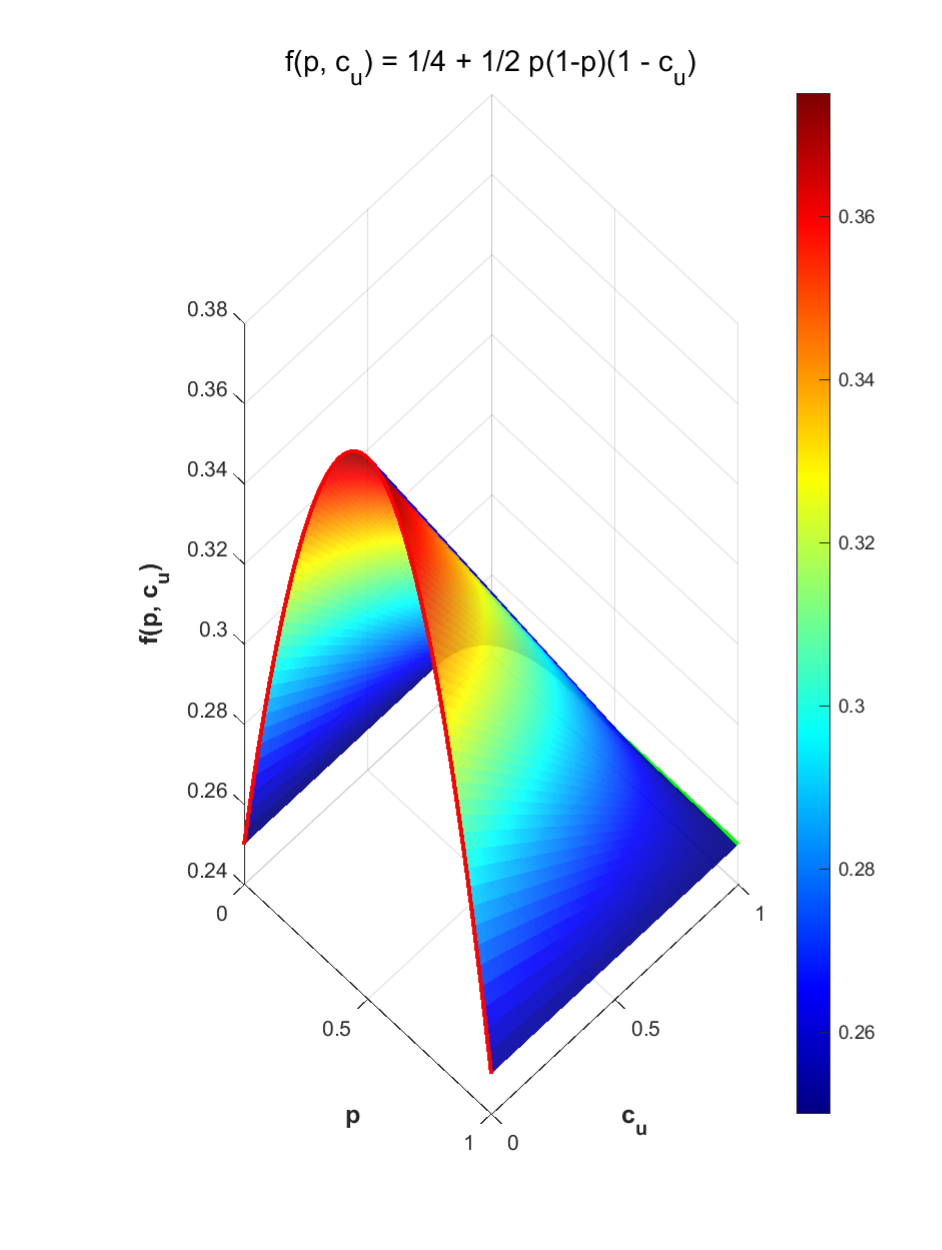}
    \subfloat[]{\label{fig:left}}
  \end{minipage}%
  \begin{minipage}{0.45\textwidth}
    \centering
    \includegraphics[width=\linewidth]{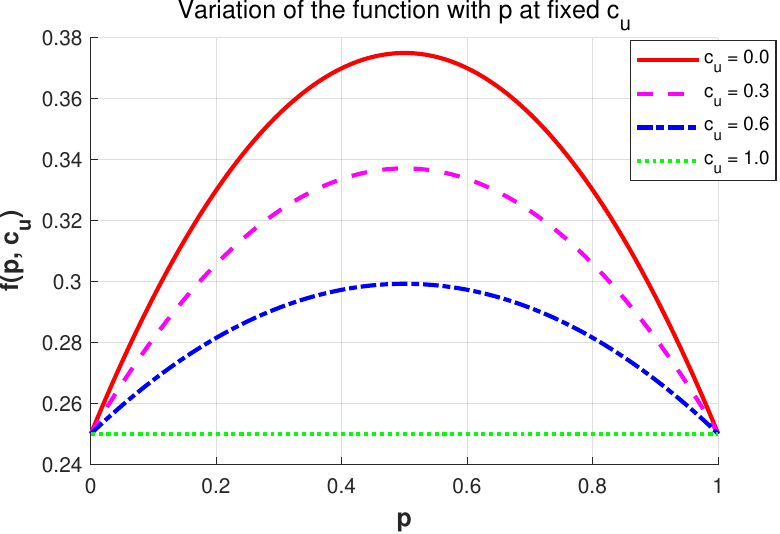}
    \subfloat[]{\label{fig:top-right}}

    \vspace{0.5em}

    \centering
    \includegraphics[width=\linewidth]{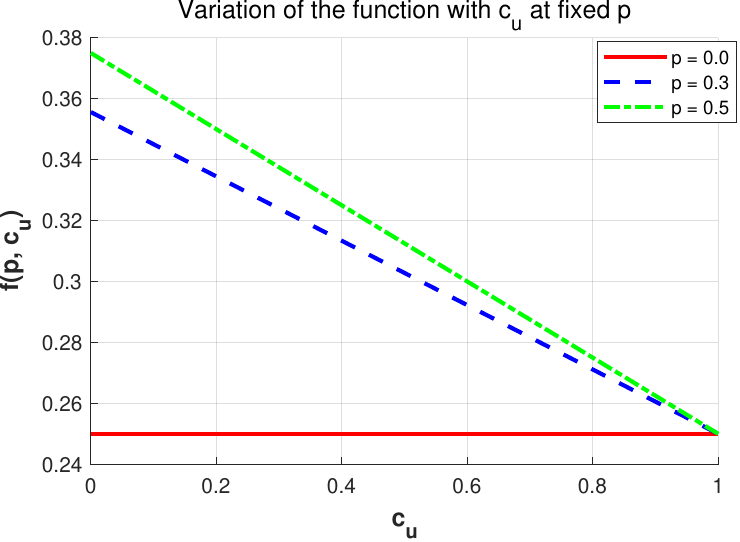}
    \subfloat[]{\label{fig:bottom-right}}
  \end{minipage}
  \caption{(Color online) A schematic illustration of the influence of parameters $p$ and $c_{u}$ on the triality relation among visibility, path distinguishability, and entanglement.
 The subfigure (a) is the plot of the Eq. (\ref{ex l}) with $c_{d}=\frac{1}{2}$. The subfigure (b) is the plot of Eq. (\ref{ex l}) about the parameter $p$ for fixed $c_{u}$.
The subfigure (c) is the plot of Eq. (\ref{ex l}) about the parameter $c_{u}$ for fixed $p$. }
  \label{tu}
\end{figure}
\end{example}

Next we illustrate this complementarity relation in Eq. (\ref{eq vde1}) by another explicit example with Werner state as the initial state.

\begin{example}
In a two-path interferometer, suppose the quantum state $\rho_{AB}$ shared between the referred particle and quantum memory is Werner state
\begin{equation}
\rho_{AB}=p |\psi^{-}\rangle_{AB}\langle\psi^{-}|+\frac{(1-p)}{4}I_{A}\otimes I_{B},
\end{equation}
where $|\psi^{-}\rangle=\frac{1}{\sqrt{2}}(|01\rangle-|10\rangle)$ and the parameter $p$ satisfies $0 \leq p \leq 1$.

Now we suppose the overlap of detector states is $|\langle d_{1}|d_{2}\rangle|=x$.
By calculation, we get that the concurrence of  $\rho_{AB}$ is $C(\rho_{AB})=\max\{0,\frac{3p-1}{2}\}$ and the distinguishability is $(P_s({q_i, |d_i\rangle_D})-\frac{1}{2})^2=\frac{1}{4}-\frac{1}{4}x^{2}$. But the visibility vanishes. Consequently, the left hand side of Eq. (\ref{eq vde1}) becomes
\begin{equation}\label{exa mix l}
    \frac{5}{16}-\frac{1}{4}x^{2}+\frac{9p^{2}-6p}{16},
\end{equation}
while the right hand side is
\begin{equation}\label{exa mix r}
    \frac{5}{8}-\frac{p^{2}(1+2x^{2})}{8}.
\end{equation}
The behaviors of the triality with respect to the parameters $p$ and $x$ are plotted in figure \ref{fg3}.
\begin{figure}[htb]
\centering
    \includegraphics[width=\textwidth]{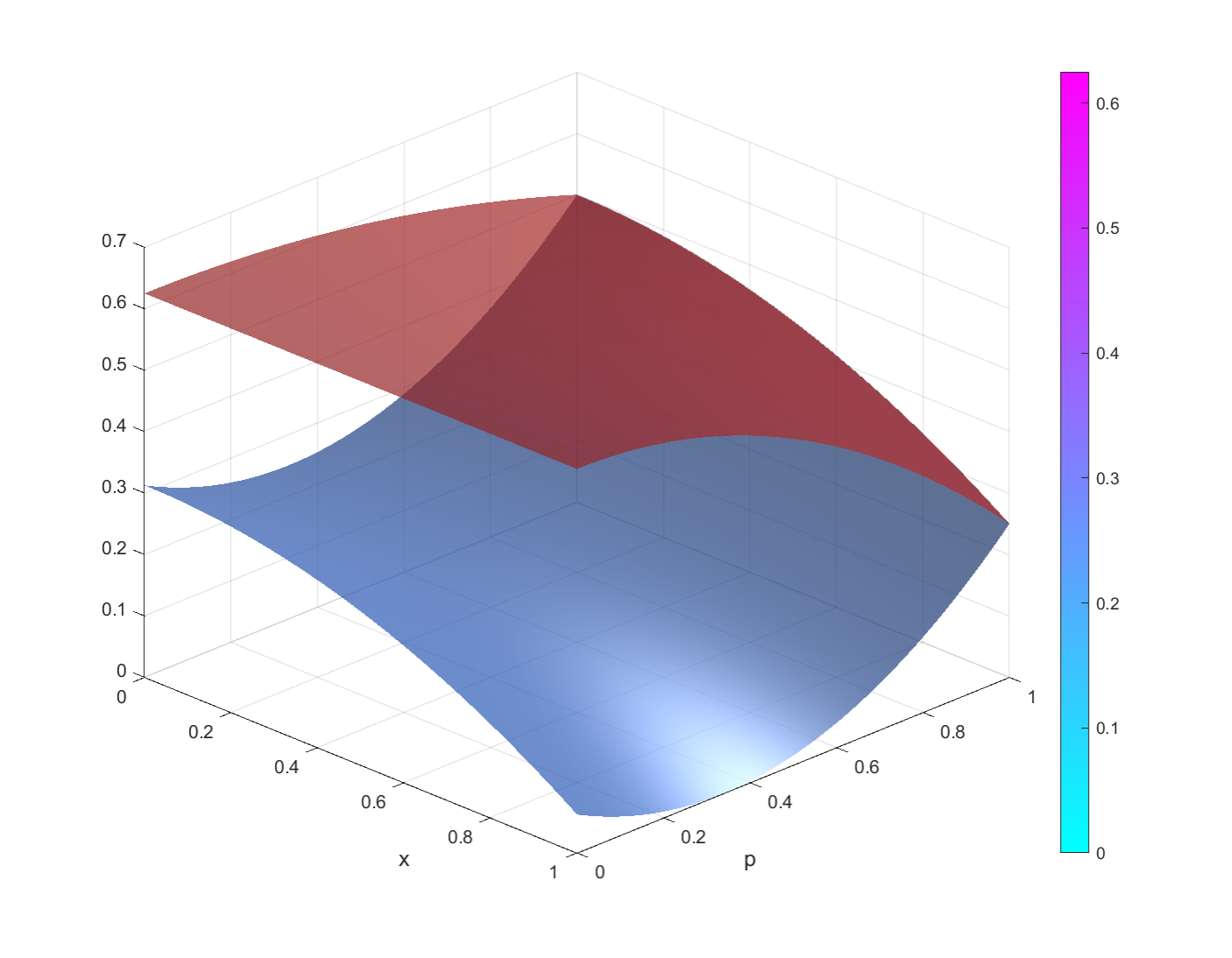}
    \caption{(Color Online) The triality relation among the visibility, path distinguishability, and entanglement with Werner state as the initial state. The surfaces above is the plot of Eq. (\ref {exa mix l}) and the surfaces below is the plot of Eq. (\ref{exa mix r}).}
    \label{fg3}
\end{figure}

Specifically if we fix the overlap of detector states as $x=\frac{\sqrt{2}}{2}$, the behaviors of the triality in Eq. (\ref {eq vde1}) with respect to the parameter $p$ is plotted in figure \ref{fg2-1}. In this case, for the entangled Werner state ($p\geq$ 1/3), Eq. (\ref {exa mix l}) approaches Eq. (\ref {exa mix r}) when the entanglement increases. Especially, when $p=1$, Eqs. (\ref {exa mix l}) and (\ref {exa mix r}) coincide. In another word, when $\rho_{AB}$
 is a maximally entangled state, we have the equality of Eq. (\ref {eq vde1}) holds.
If we fix the parameter $p=\frac{3}{4}$, the behaviors of the triality with respect to the parameter $x$ is plotted in figure \ref{fg2-2}.
In this case, the gap between the right hand side and the left hand side of Eq. (\ref {eq vde1}) increases as $x$ increases.

\begin{figure}[ht]
    \centering
    \begin{minipage}{0.45\textwidth}
        \centering
        \includegraphics[width=\textwidth]{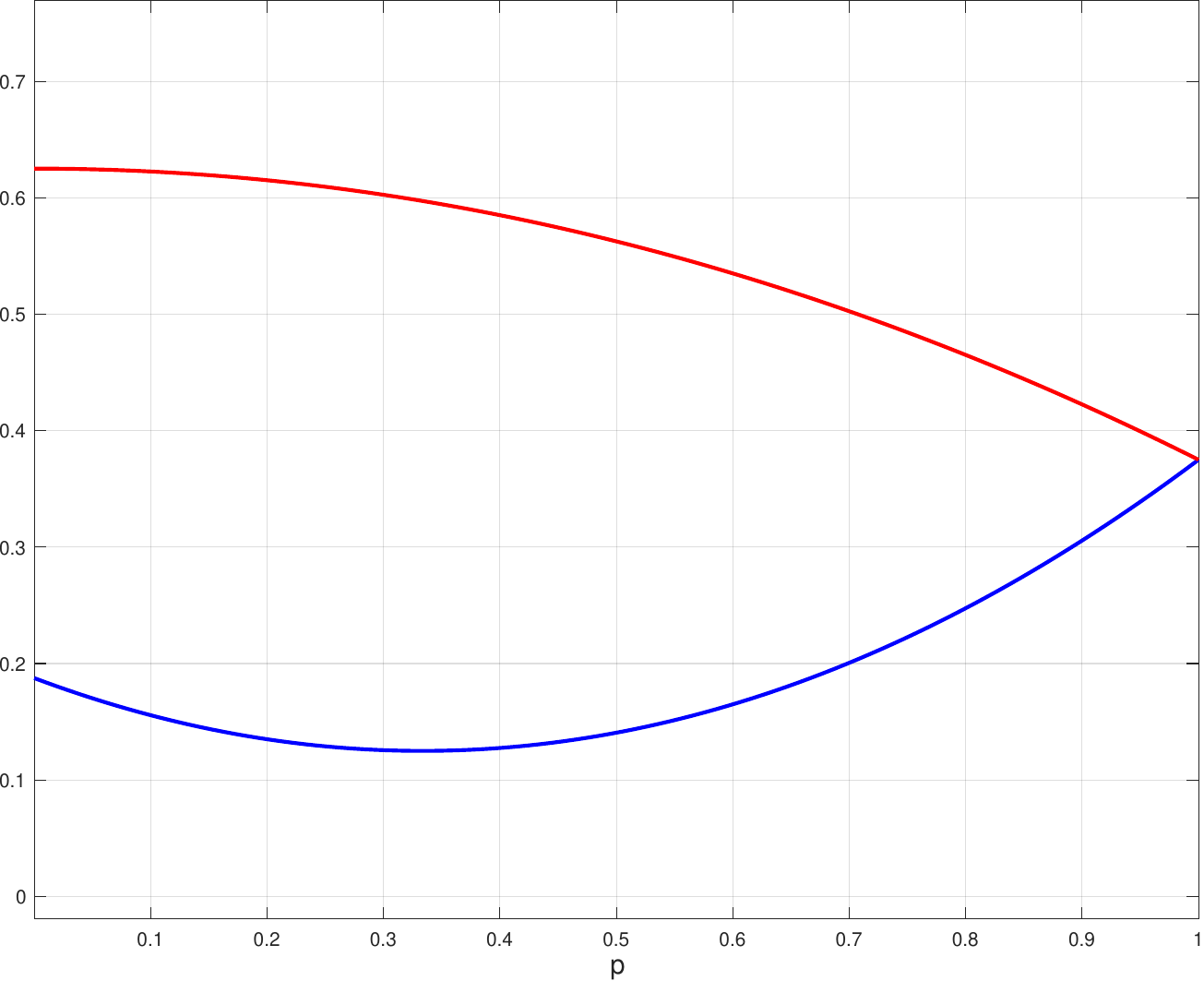} 
        \caption{(Color Online) The estimation of the triality relation among visibility, path distinguishability, and entanglement with respect to the parameter $p$ for the fixed overlap $|\langle d_1|d_2\rangle|=\frac{\sqrt{2}}{2}$. The red curve above corresponds to Eq. (\ref {exa mix r}), and the blue curve below to Eq. (\ref {exa mix l}).}
        \label{fg2-1}
    \end{minipage}\hfill
    \begin{minipage}{0.5\textwidth}
        \centering
        \includegraphics[width=\textwidth]{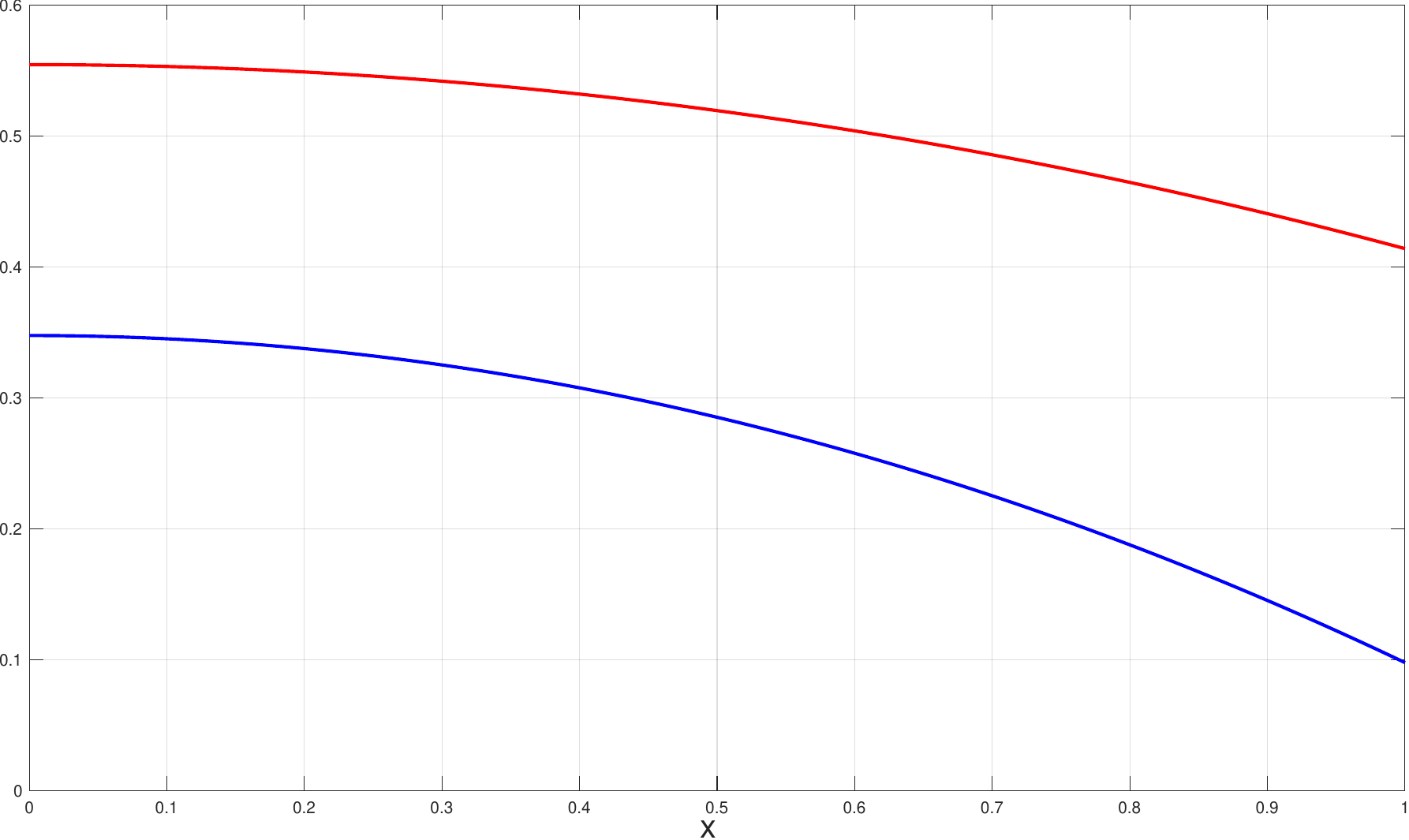} 
        \caption{(Color Online) The estimation of the triality relation among visibility, path distinguishability, and entanglement with respect to the parameter $x=|\langle d_1|d_2\rangle|$ for the fixed $p=3/4$. The red curve above corresponds to  Eq. (\ref {exa mix r}), while the blue curve below corresponds to Eq. (\ref {exa mix l}).}
        \label{fg2-2}
    \end{minipage}
\end{figure}
\end{example}

\section{Conclusions}

In a multi-path interferometer equipped with detectors and a quantum memory, we have presented the duality relation between the visibility and the path distinguishability, and the triality among the visibility, the path distinguishability, and mixedness respectively. Especially for the two-path interferometer, the corresponding results are complete. We have also derived the triality relation among the visibility, the path distinguishability, and entanglement, which reveals the role of entanglement in the multi-path interferometer with quantum memory quantitatively.

Here it is worth to stress that the similar duality and triality can also be derived if we replace the visibility by the $l_p$-norm coherence, since the $l_p$-norm is less than or equal to the $l_2$-norm for $1< p < \infty$. So the complementarity relations we proved in this paper are typical in some sense.

Meanwhile, some similar complementarity relations in the multi-path interferometer have been proposed in the literature. The difference between our results and the previous ones is not the explicit measures, but the physical scenario. For example,
the triality relation among predictability, visibility, and mixedness in \cite{hh1,n3} is true in the absence of detectors. And the triality relations involving visibility, predictability, and entanglement in  \cite{JC3,Ding} is true without quantum memory as well as the input state is pure.
So these complementarity relations apply to different physical scenarios.

\begin{ack}{\rm We thank the anonymous referees for their valuable comments and suggestions. M. J. Zhao thanks the center for Quantum Information,
Institute for Interdisciplinary Information Sciences,
Tsinghua University for hospitality. This work is partially supported by the National Natural Science Foundation of China (Nos. 11671201 and 12171044) and Postgraduate Research \& Practice Innovation Program of Jiangsu Province (KYCX25\_0627).}
\end{ack}

\end{document}